\documentclass[%
 reprint,
 amsmath,amssymb,
 aps,
]{revtex4-1}

\usepackage{amsmath,amssymb,amstext,dsfont,mathtools}
\usepackage{graphicx}
\usepackage{dcolumn}
\usepackage{bm}
\usepackage{scrextend}
\usepackage{xparse}
\usepackage[standard, thref, amsmath, thmmarks]{ntheorem}

\usepackage[pdftex,pagebackref=false]{hyperref} % with basic options
\hypersetup{
    plainpages=false,       % needed if Roman numbers in frontpages
    unicode=false,          % non-Latin characters in Acrobat’s bookmarks
    pdftoolbar=true,        % show Acrobat’s toolbar?
    pdfmenubar=true,        % show Acrobat’s menu?
    pdffitwindow=false,     % window fit to page when opened
    pdfstartview={FitH},    % fits the width of the page to the window
    pdftitle={Pauli Partitioning: in Theory and in Practice},
    pdfauthor={Andrew Jena, Scott Genin, Michele Mosca},
    pdfsubject={Combinatorics and Optimization (Quantum Information)},
    pdfnewwindow=true,      % links in new window
    colorlinks=true,        % false: boxed links; true: colored links
    linkcolor=blue,         % color of internal links
    citecolor=green,        % color of links to bibliography
    filecolor=magenta,      % color of file links
    urlcolor=cyan           % color of external links
}

\newcommand{\defeq}{\vcentcolon=}

\newcommand{\bra}[1]{\left.\left\langle#1\right.\right|}
\newcommand{\ket}[1]{\left.\left|#1\right.\right\rangle}

\renewcommand{\mod}[1]{\ (\textrm{mod }#1)}

\newcommand{\identity}[1]{\mathds{1}_{#1}}

\newcommand{\SUM}{SU\hspace{-0.2em}M}
\newcommand{\reducesto}{\leq_{\textrm{P}}}
\newcommand{\equivto}{\equiv_{\textrm{P}}}
\newcommand{\kpart}{\textrm{$k$Part} }
\newcommand{\kcolor}{\textrm{$k$Color} }
\newcommand{\symplectic}[2]{\left(\hspace{-0.4em}\begin{tabular}{c|c} $\begin{matrix}#1\end{matrix}$ & $\begin{matrix}#2\end{matrix}$ \end{tabular}\hspace{-0.4em}\right)}
\newcommand{\diag}[1]{Diag_{#1}}
\DeclareDocumentCommand{\pauli}{m o}  
  {\IfValueTF{#2}{\mathcal{P}_{#1}^{#2}}{\mathcal{P}_{#1}}
  }
\DeclareDocumentCommand{\Zmod}{m o}  
  {
    \IfValueTF{#2}{\left(\mathbb{Z}/#1\mathbb{Z}\right)^{{#2}}}{\mathbb{Z}/#1\mathbb{Z}}
  }

\begin{document}

\newtheorem{conjecture}{Conjecture}

%%%%%%%%%%%%%%%%
% FRONT MATTER %
%%%%%%%%%%%%%%%%

% FRONT MATTER

\title{Pauli Partitioning with Respect to Gate Sets}%

\author{Andrew Jena}
\email{ajjena@uwaterloo.ca}
\affiliation{%
Department of Combinatorics and Optimization, University of Waterloo
}%
\affiliation{%
Institute for Quantum Computing, University of Waterloo
}%

\author{Scott Genin}
\email{scott.genin@otilumionics.com}
\affiliation{
OTI Lumionics, Toronto, ON
}%

\author{Michele Mosca}
\email{michele.mosca@uwaterloo.ca}
\affiliation{%
Department of Combinatorics and Optimization, University of Waterloo, Waterloo, ON
}%
\affiliation{%
Institute for Quantum Computing, University of Waterloo, Waterloo, ON
}%
\affiliation{%
Perimeter Institute for Theoretical Physics, Waterloo, ON
}%

\date{\today}%

\begin{abstract}
Measuring the expectation value of Pauli operators on prepared quantum states is a fundamental task in a multitude of quantum algorithms. Simultaneously measuring sets of operators allows for fewer measurements and an overall speedup of the measurement process. We investigate the task of partitioning a random subset of Pauli operators into simultaneously-measurable parts. Using heuristics from coloring random graphs, we give an upper bound for the expected number of parts in our partition. We go on to conjecture that allowing arbitrary Clifford operators before measurement, rather than single-qubit operations, leads to a decrease in the number of parts which is linear with respect to the lengths of the operators. We give evidence to confirm this conjecture and comment on the importance of this result for a specific near-term application: speeding up the measurement process of the variational quantum eigensolver.
\end{abstract}


\maketitle

%\tableofcontents

%%%%%%%%%%%%%%%%
% MAIN MATTER  %
%%%%%%%%%%%%%%%%

%----------------------------------------------------------------------
\section{Introduction}\label{sec:1}
%----------------------------------------------------------------------

Our motivation for partitioning Pauli operators was, and remains, the speedup of the measurement step of the variational quantum eigensolver. However, upon recognizing that our results were generalizable, we chose to frame our results for qudit operators of prime dimension, $q$, which can be replaced by 2 for most near-term applications. We begin by establishing notation for the remainder of the paper.

\begin{definition}
We shall use the following generalization of the Pauli operators, which are often referred to as the shift and clock operators, respectively. For a prime, $q$, we define the following $q \times q$ unitary matrices:
\[
X_q = \begin{pmatrix}
0 & 0 & \dotso & 0 & 1 \\
1 & 0 & \dotso & 0 & 0 \\
0 & 1 & \dotso & 0 & 0 \\
\vdots & \vdots & \ddots & \vdots & \vdots \\
0 & 0 & \dotso & 1 & 0
\end{pmatrix};\ \ \ 
Z_q = \begin{pmatrix}
1 & 0 & \dotso & 0 \\
0 & \omega_q & \dotso & 0 \\
\vdots & \vdots & \ddots & \vdots \\
0 & 0 & \dotso & \omega_q^{q-1}
\end{pmatrix},
\]
where $\omega_q = e^{2\pi i/q}$. For $q = 2$, these are recognizable as the $2 \times 2$ Pauli matrices.
\end{definition}

\begin{definition}
Let $\pauli{q}$ denote the generalized Pauli group (ignoring phases) over $\Zmod{q}$. I.e., we define:
\[
\pauli{q} \defeq \left\{ X_q^i Z_q^j : i,j \in \Zmod{q} \right\}.
\]
Similarly, we define:
\[
\pauli{q}[n] = \left\{ \bigotimes_{i=1}^n P_i : P_i \in \pauli{q} \right\}
\]
to be the set of generalized length-$n$ Pauli operators (still ignoring phases) over $\Zmod{q}$.
\end{definition}

Since our task is to partition Pauli operators into parts within which all the operators can be simultaneously measured, it is natural to first investigate what this means. If we are able to simultaneously diagonalize the operators in a part, then we may perform measurements in the computational basis and simultaneously measure all the operators. For a chosen gate set, we may partition our operators into \emph{diagonalizable} parts, i.e., parts in which all operators are simultaneously-diagonalizable by an element of the gate set.

\begin{definition}
Given a gate set, $GS$, we shall denote by $\diag{GS}$ the set of sets of operators which are simultaneously-diagonalizable by some element of $GS$.
\end{definition}

In Section~\ref{sec:2}, we tackle the problem of partitioning a set of arbitrary-length Pauli operators, proving that this problem is equivalent to graph coloring, and is therefore NP-hard in general. In Section~\ref{sec:3}, we use this equivalence to compare the expected number of parts in partitions for two gate sets: arbitrary Clifford operators versus single-qudit Clifford operators. We conjecture that, when arbitrary Clifford operators are allowed, the number of parts is decreased by a factor linear with respect to the length of the Pauli operators. In Section~\ref{sec:4}, we discuss the variational quantum eigensolver, our motivation for this research, and give some evidence to support our conjecture. In Section~\ref{sec:5}, we conclude by summarizing our results and discussing some open problems for further research.

%----------------------------------------------------------------------
\section{Partitioning Sets of Pauli Operators}\label{sec:2}
%----------------------------------------------------------------------

Given a set of Pauli operators and a given gate set, how many parts might we expect to be optimal, and how might we go about finding such an optimal partition? Letting $\pauli{q}[*]$ denote the set of Pauli operators of dimension $q$ and arbitrary length, we shall formalize this problem and show its equivalence to graph coloring. We begin by defining some notation for our reductions and the relevant problems.

\begin{definition}
A decision problem, $A$, polytime reduces to another, $B$, ($A \reducesto B$) if there exists a polytime algorithm which solves $A$ given an oracle for solving $B$.
\end{definition}

\begin{definition}
$A$ is polytime equivalent to $B$ ($A \equivto B$) if $A \reducesto B$ and $B \reducesto A$.
\end{definition}

\begin{definition}
A partition of $\mathcal{S} \subseteq \pauli{q}[*]$ into $k$ diagonalizable parts is a \textbf{$\bm k$-partition} of $\mathcal{S}$.
\end{definition}

\begin{definition}
We define the $k$-partitioning problem (\kpart\hspace{-0.4em}) as follows:
\begin{addmargin}[2.2em]{0em}
\begin{itemize}
\item[Given:] $\mathcal{S} \subseteq \pauli{q}[*]$ and $k \in \mathbb{Z}_{\geq 1}$
\item[Question:] does there exists a $k$-partition of $\mathcal{S}$?
\end{itemize}
\end{addmargin}
\end{definition}

\begin{definition}
A \textbf{$\bm k$-coloring} of a simple, undirected graph, $G$, is a partition of the vertices of $G$ into $k$ co-cliques.
\end{definition}

\begin{definition}
We define the $k$-coloring (\kcolor\hspace{-0.4em}) problem as follows:

\begin{addmargin}[2.2em]{0em}
\begin{itemize}
\item[Given:] a simple, undirected graph, $G$, and $k \in \mathbb{Z}_{\geq 1}$
\item[Question:] does there exists a $k$-partition of $\mathcal{S}$?
\end{itemize}
\end{addmargin}
\end{definition}

\begin{proposition}
$\kpart \equivto \kcolor$
\end{proposition}
\begin{proof}
We shall prove the reduction in both directions.
\begin{itemize}
\item $\kpart \reducesto \kcolor$: let $O^C$ be an oracle for \kcolor.

Given a set, $\mathcal{S} \subseteq \pauli{q}[*]$, and an integer, $k \in \mathbb{Z}_{\geq 1}$, we construct a graph, $G = (V,E)$, (we'll call this the \textbf{non-diagonalizable graph}) by letting:
\begin{itemize}
\item $V = \mathcal{S}$
\item $E = \left\{ P_iP_j : \begin{array}{l}
P_i, P_j \in \mathcal{S} \textrm{ and }\\
\{P_i,P_j\} \not\in \diag{GS}
\end{array} \right\}.$
\end{itemize}

We observe that a $k$-coloring of $G$ is a partition of its vertices into co-cliques. In other words, each part is a diagonalizable set of Pauli operators. Thus, a $k$-partition of $\mathcal{S}$ exists if and only if a $k$-coloring of $G$ exists.

Therefore, the output of $O^C(G,k)$ is a solution to $\kpart$.

\item $\kcolor \reducesto \kpart$: let $O^P$ be an oracle for \kpart.

Given a graph, $G = (V,E)$, and an integer, $k \in \mathbb{Z}_{\geq 1}$, we first define the adjacency matrix of $G$ to be the matrix, $A_G$, such that:
\[
A_G(u,v) = \begin{cases}
1 & uv \in E \\
0 & uv \not\in E
\end{cases}
\]
We then construct a set $\mathcal{S} \subseteq \pauli{q}[*]$ by letting:
\begin{itemize}
\item $\mathcal{S} = \symplectic{\identity{n}}{A_G^{LT}}$,
\end{itemize}
where $A_G^{LT}$ is the lower triangular portion of the adjacency matrix of $G$. Indexing our set of Pauli operators by the vertices of $G$, we observe that $\{P_u, P_v\} \in \diag{GS}$ if and only if $uv \not\in G$.

We observe that a $k$-partition of $\mathcal{S}$ is a partition of the Pauli operators into diagonalizable parts. In other words, each part is a co-clique of the vertices of $G$. Thus, a $k$-coloring of $G$ exists if and only if a $k$-partition of $\mathcal{S}$ exists.

Therefore, the output of $O^P(\mathcal{S},k)$ is a solution to \kcolor.
\end{itemize}
\end{proof}

It has long been known that \kcolor (for $k \geq 3$) is NP-complete with respect to $|V(G)|$. Since the number of vertices in $G$ is the same as the number of Pauli operators in our set, $\mathcal{S}$, in both directions of our above equivalence, this immediately gives us the following result.

\begin{theorem}
\kpart (for $k \geq 3$) is NP-complete with respect to $|\mathcal{S}|$.
\end{theorem}

Using a similar reduction to above, we may show that the following problem is NP-hard.

\begin{definition}
Given $\mathcal{S} \subseteq \pauli{q}[*]$, the \textbf{Pauli partitioning problem} is to return a partition of $\mathcal{S}$ into the fewest number of diagonalizable parts.
\end{definition}

In fact, this problem is exactly equivalent to coloring the corresponding non-diagonalizable graph with the fewest colors. As such, we can say a few things about the expected number of parts in a partition of a randomly chosen set of Paulis from $\pauli{q}[*]$.

%----------------------------------------------------------------------
\section{Comparing Partition Sizes for Certain Gate Sets}\label{sec:3}
%----------------------------------------------------------------------

Choosing a gate set for a given application of these partitioning techniques has a lot of factors. The first and foremost is the capabilities of the device being used. If certain gates are not reliable (for instance, entanglement gates), more-restrictive gates sets might be necessary for obtaining useful outputs. Another consideration, however, should be the expected number of parts in a partition.

Below, we compare the expected number of parts in a partition with respect to two different gate sets, the generalized (for qudits) Clifford group, $\mathcal{C}$, and the set of all single-qudit Clifford operators, $sq\mathcal{C}$.

Beginning with the generalized Clifford group, we first want to investigate which sets are in $\diag{\mathcal{C}}$. It is well known that any set of commuting operators is simultaneously-diagonalizable by a unitary \cite{Diag}. In Appendix~\ref{app:A}, we give an explicit construction to show that any set of commuting Pauli operators is simultaneously-diagonalizable by a Clifford operator. This allows us to give some estimates for the expected number of parts in our partitions.

\begin{proposition}\label{prop:chromatic}
For almost all sets, $\mathcal{S} \subseteq \pauli{q}[*]$, such that all the operators in $\mathcal{S}$ are linearly independent, the number of parts in a minimal partition of $\mathcal{S}$ with respect to $\mathcal{C}$ is:
\[
\left( \frac{1}{2} + o(1) \right) \left( \log_2(q) + o(1) \right) \frac{|\mathcal{S}|}{\log_2(|\mathcal{S}|)}
\]
\end{proposition}
\begin{proof}
Let $m$ be the largest length of any Pauli operator in $\mathcal{S}$.

For any two randomly-chosen Pauli operators, $P_i, P_j \in \pauli{q}[m] \setminus \identity{q^m}$, the probability that $\{P_i,P_j\} \in \diag{\mathcal{C}}$ is $(q^{2m-1}-2)/(q^{2m}-1)$. If we choose a single non-identity qudit of $P_i$, then $P_j$ may take any values on the remaining $m-1$ qubits (i.e. $q^{2m-2}$ possibilities), but must take any one of $q$ values on the last qudit to make the commutator 0 (i.e. $q^{2m-1}$ possibilities), and we subtract 2 for the identity operator and for $P_i$, itself.

Choosing an ordering on our set, $\mathcal{S} = \{P_1,\dotso, P_{|S|}\}$, we define the following matrix:
\[
C(i,j) = P_i \odot P_j,
\]
where $\odot$ is the symplectic inner product (which is analogous to the commutator, but with outputs in $\Zmod{q}$). We observe that we may always find a Clifford group operation which acts by conjugation on our set transforming it into:
\[
\symplectic{\identity{|\mathcal{S}|}\ 0_{n-|\mathcal{S}|}}{C^{LT}\ 0_{n-|\mathcal{S}|}}.
\]
This is because the commutator between any pair of operators is preserved when we conjugate by a Clifford operator. Thus, for any operator outside of $\mathcal{S}$, whether it commutes with $P_i$ is determined solely by the power of the $X$ term on its $i$\textsuperscript{th} qubit. Since the powers on each qudit are chosen independently, the probabilities are independent.

This implies that the non-diagonalizable graph is a random graph on $|\mathcal{S}|$ vertices with edge-probability $1-(q^{2m-1}-2)/(q^{2m}-1)$. Using the result from \cite{GraphColor}, the minimum number of colors necessary to color $G$ is expected to be:
\begin{align*}
\left( \frac{1}{2} + o(1) \right) &\log_2 \left( \frac{1}{1-\left( 1 - \frac{q^{2m-1}-2}{q^{2m}-1} \right)} \right) \frac{|\mathcal{S}|}{\log_2(|\mathcal{S}|)} \\
= \left( \frac{1}{2} + o(1) \right) &\log_2 \left( \frac{q^{2m}-1}{q^{2m-1}-2} \right) \frac{|\mathcal{S}|}{\log_2(|\mathcal{S}|)} \\
= \left( \frac{1}{2} + o(1) \right) &\left( \log_2(q) + o(1) \right) \frac{|\mathcal{S}|}{\log_2(|\mathcal{S}|)}
\end{align*}
\end{proof}

In the last proposition, we have clearly required the assumption that the operators in our set were linearly independent. However, we conjecture that the value found above remains an upper bound for an arbitrary set, $S$.

\begin{conjecture}\label{conj:chromatic}
For almost all sets, $\mathcal{S} \subseteq \pauli{q}[*]$, the number of parts in a minimal partition of $\mathcal{S}$ with respect to $\mathcal{C}$ is bounded above by:
\[
\left( \frac{1}{2} + o(1) \right) \left( \log_2(q) + o(1) \right) \frac{|\mathcal{S}|}{\log_2(|\mathcal{S}|)}
\]
\end{conjecture}
\begin{proof}
First, we give evidence for why we might expect our assumption about linear independence to be satisfied for a randomly-chosen set.

The probability that there are no linearly dependent subsets of a set, $\mathcal{S}$, of length-$m$, $q$-ary Pauli operators is:
\[
\prod_{i=0}^{|\mathcal{S}|-1} \left( 1-q^{i-2m} \right) = (q^{-2m};q)_{|\mathcal{S}|},
\]
where $(a;q)_k$ is the $q$-Pochhammer symbol.

For values of $|\mathcal{S}| \leq 2m$, this probability is bounded below by a fixed value which depends on $q$. For $q=2$, it is bounded by $\approx 0.288788$ and for larger values of $q$, this value tends towards $1$.

In other words, for any value of $q$, this size requirement is enough to ensure our graph is random with probability greater than $1/4$.

However, when $|\mathcal{S}| > 2m$, we must rely on a different argument which gives a heuristic for why the number of partitions should actually be lower than if there were not linear dependence.

In instances where we have linear dependence in our Pauli operators, we must observe that linearly dependent sets lead to larger co-cliques in our graph. If some set of Pauli operators, $S$, is in $\diag{\mathcal{C}}$, then the span of $S$ (i.e. the set of operators which can be written as a product over the operators in $S$) is also in $\diag{\mathcal{C}}$. This is because any operator which commutes with a set of operators will still commute with a product of operators from that set. Since we would expect to have larger co-cliques than in a random graph, our graph should admit a coloring with fewer colors.
\end{proof}

Having established these expectations for the generalized Clifford group, we can now compare these results with what we might expect from the set of single-qudit Clifford operators, $sq\mathcal{C}$.

Given the gate set, $sq\mathcal{C}$, we observe that any two operators which do not commute on a given qudit  are not simultaneously-diagonalizable. This is because conjugation by a Clifford operator maintains commutation relations, and without entanglement gates, we are unable to make these operators commute on the given bit. On the other hand, if two operators commute on a given bit, we may diagonalize that qudit by some single-qudit Clifford operator. This allows us to simultaneously-diagonalize any quditwise-commuting Pauli operators using an element of $sq\mathcal{C}$.

\begin{conjecture}
Given $\mathcal{S} \subseteq \pauli{q}[*]$, we expect the number of parts in a minimal partition of $\mathcal{S}$ with respect to $sq\mathcal{C}$ to be bounded below by:
\[
\left( \frac{1}{2} + o(1) \right) m(\log_2(q) - o(1)) \frac{|\mathcal{S}|}{\log_2(|\mathcal{S}|)},
\]
where $m$ is the length of the largest Pauli operator in $\mathcal{S}$.
\end{conjecture}
\begin{proof}
Letting $m$ be the maximum length of any of the chosen Paulis, we first look at the probability that any two operators, $P_i$ and $P_j$, chosen from $\pauli{q}[m]$ will quditwise commute. This requires that each pair of qudits commutes, which happens with probability $(q^3+q^2-q)/q^4$. Thus, the probability that these operators quditwise commute is $((q^3+q^2-q)^m-q^{2m})/(q^4)^m$.

The non-diagonalizable graph will therefore have edge probability $1-(q^3+q^2-q)^m/(q^4)^m$, although the edge probabilities will likely not be independent. If they were independent, we would find the number of parts to be:
\begin{align*}
&\left( \frac{1}{2} + o(1)\right) \log_2 \left( \frac{1}{1-\left(1-\frac{(q^3+q^2-q)^m}{(q^4)^m}\right)} \right) \frac{|\mathcal{S}|}{\log_2(|\mathcal{S}|)} \\
= &\left( \frac{1}{2} + o(1)\right) \log_2 \left( \frac{(q^4)^m}{(q^3+q^2-q)^m} \right) \frac{|\mathcal{S}|}{\log_2(|\mathcal{S}|)}\\
= &\left( \frac{1}{2} + o(1)\right) m\log_2 \left( \frac{q^4}{q^3+q^2-q} \right) \frac{|\mathcal{S}|}{\log_2(|\mathcal{S}|)} \\
= &\left( \frac{1}{2} + o(1)\right) m \left( \log_2(q) - o(1) \right) \frac{|\mathcal{S}|}{\log_2(|\mathcal{S}|)}.
\end{align*}

Without the independence of these edge probabilities, however, it is difficult to give a concrete lower bound on the number of colors we should expect using single-qudit Clifford operators.
\end{proof}

Running with our conjectures, however, we compare this to Conjecture~\ref{conj:chromatic} and observe that the expected number of parts is roughly $\mathcal{O}(m)$ times more using bitwise commutation rather than general commutation. Since the greedy algorithm considers a partition requirement which is even stronger than bitwise commutation, graph coloring algorithms would greatly improve upon existing partitioning techniques.

We have seen these conjectured results borne out on small examples. In particular, we have used a greedy graph coloring algorithm to partition qubit Pauli operators and have compared the number of parts given by each gate set for some small examples. In the table below, the length of the Pauli operators is shown alongside the ratio of the number of parts given by the gate set, $sg\mathcal{C}$, versus the gate set, $\mathcal{C}$. The larger this ratio is, the larger the improvement we obtain from expanding our gate set.

\begin{center}
\begin{tabular}{|c|c|}
\hline
\textbf{Length} & \textbf{Ratio} \\ \hline
1 & 1 \\ \hline
2 & 1.483 \\ \hline
3 & 1.575 \\ \hline
4 & 1.930 \\ \hline
5 & 2.389 \\ \hline
6 & 2.948 \\ \hline
7 & 3.178 \\ \hline
8 & 3.709 \\ \hline
9 & 4.119 \\ \hline
10 & 4.793 \\ \hline
11 & 5.278 \\ \hline
12 & 5.793 \\ \hline
13 & 6.376 \\ \hline
14 & 6.906 \\ \hline
15 & 7.501 \\ \hline
16 & 8.043 \\ \hline
17 & 8.545 \\ \hline
18 & 9.054 \\ \hline
19 & 9.529 \\ \hline
\end{tabular}
\end{center}

On these small examples, which were obtained by averaging over sample sizes of only 5 sets, the increase in the ratio grows with respect to $\mathcal{O}(length)$. While this in no way proves the above conjectures, it gives some evidence that the improvements gained by expanding our gate set is non-negligible.

%----------------------------------------------------------------------
\section{Measurement in the Variational Quantum Eigensolver}\label{sec:4}
%----------------------------------------------------------------------

The variational quantum eigensolver is a quantum-classical hybrid algorithm used for finding the ground state energy of a molecule. With applications ranging from quantum chemistry to combinatorial optimization and a low cost in quantum resources, the variational quantum eigensolver is a great candidate for near-term applications of quantum computers.

\begin{definition}
The \textbf{variational quantum eigensolver} solves the following problem:
\begin{addmargin}[2.2em]{0em}
\begin{itemize}
\item[Given:] $H = \sum_k c_kP_k$, a Hamiltonian written as a sum over Pauli operators
\item[Goal:] approximate the smallest eigenvalue, $\lambda$, of $H$.
\end{itemize}
\end{addmargin}
\end{definition}

To begin our analysis, we see that, if we were able to initialize the eigenstate, $\ket{\psi}$, such that $H\ket{\psi} = \lambda \ket{\psi}$, and if we could measure $\bra{\psi}H\ket{\psi} = \bra{\psi}\lambda\ket{\psi} = \lambda$, then we could complete our goal. However, there are two problems with this:
\begin{enumerate}
\item How do we produce $\ket{\psi}$?
\item How do we measure $\bra{\psi}H\ket{\psi}$?
\end{enumerate}

The first of the above problems is dealt with by a classical optimization algorithm. Preparing an initial state, $\ket{\psi_0}$, and measuring $\bra{\psi_0}H\ket{\psi_0}$, we plug our results into a classical optimizer which returns parameters for a new state, $\ket{\psi_1}$. Continuing this process, we improve our state until we construct some $\tilde{\ket{\psi}}$ such that $|\tilde{\bra{\psi}}H\tilde{\ket{\psi}} - \lambda| \leq \varepsilon$, for some desired precision.

The second problem, however, seemingly requires us to simulate our Hamiltonian, $H$, which may be very difficult. To get around this, we observe that expectation values are linear, so we have:
\[
\bra{\psi}H\ket{\psi} = \left\langle\psi\left|\sum c_k P_k \right|\psi\right\rangle = \sum c_k \bra{\psi} P_k \ket{\psi}.
\]
Measuring $\bra{\psi} P_k \ket{\psi}$ is as simple as transforming $P_k$ into a diagonal Pauli operator and measuring in the computational basis. Therefore, we are able to accomplish this step efficiently.

While the algorithm we have outlined above looks feasible at first glance, it is still not practical for many near-term quantum devices. The first issue is the length of these Pauli operators (which determines the number of qubits in the system). These scale proportionally to the size of the molecules being considered, and this quickly puts many interesting molecules out of the reach of near-term quantum devices.

The second issue, however, is the sheer number of runs required to make the algorithm work. Considering we require this many measurements (and this many initializations of $\ket{\psi_i}$), in each quantum step of our algorithm, it is important to try to cut down on the number of measurements required. This was our motivation for the previous discussion on partitioning Pauli operators.

We used a greedy graph coloring algorithm to partition the Pauli operators of the Hamiltonians for various molecules. We did so with respect to both the generalized Clifford group and the set of single-qubit Clifford operators. We summarize our results in the table below.

\begin{center}
\begin{tabular}{|c|c|c|c|c|}
\hline
\textbf{Molecule} & \textbf{Length} & \textbf{\# Paulis} & $\bm{sq\mathcal{C}}$ & $\bm{\mathcal{C}}$ \\ \hline
H$_2$ & 2 & 11 & 9 & 5 \\ \hline
H$_2$ & 4 & 15 & 3 & 2 \\ \hline
LiH & 4 & 27 & 8 & 4 \\ \hline
LiH & 6 & 118 & 36 & 13 \\ \hline
H$_2$O & 8 & 197 & 42 & 15 \\ \hline
phenyl group & 32 & 20481 & 7826 & 270 \\ \hline
\end{tabular}
\end{center}

As we see, the improvements afforded by changing our gate set allow for significantly fewer preparations and measurements of the expectation value. This improvement allows for the simulation of larger Hamiltonians in near term implementations, provided the fidelity of entanglement gates can be ensured.

%----------------------------------------------------------------------
\section{Conclusion and Open Problems}\label{sec:5}
%----------------------------------------------------------------------

In this paper, we have covered a lot of ground with regard to partitioning Pauli operators. In particular, we have shown that minimal partitions of $\pauli{q}[n]$ always exist, leading to a proof of the existence of a set of $q^n+1$ mutually unbiased bases in $\mathbb{C}^{q^n}$. Moreover, we have shown that the problem of partitioning Pauli operators is equivalent in complexity to the problem of coloring graphs, and that implementing graph coloring algorithms should significantly reduce the number of measurements required in each quantum step of the variational quantum eigensolver. We provide some launching off points for future research for those readers interesting.

First, another approach to speeding up the variational quantum eigensolver would be to find a similar matrix with a more simple expression as a sum over Pauli operators. For instance, if one could efficiently calculate the diagonalization of $H$ (i.e. write a matrix similar to $H$ as a sum over all diagonal Pauli matrices), then all measurements could be performed simultaneously, saving a significant amount of time. This, or similar approaches, would allow for significantly faster and more reliable computations. The variational quantum state diagonalization \cite{VQSD} algorithm does exactly this, but it already requires more quantum resources than the variational quantum eigensolver. Perhaps a hybrid algorithm which mostly diagonalizes $H$ before applying the variational quantum eigensolver could require fewer quantum resources.

Second, while we showed the equivalence between partitioning Pauli operators and coloring graphs, we did not make mention of specific graph coloring algorithms. This is because, as was mentioned many a time through Sections	~\ref{sec:2} and~\ref{sec:3}, the graphs are not truly random. In fact, the set of Pauli operators, itself, is not random, since it is the output of an algorithm for mapping electronic states onto qudits (e.g. the Bravyi-Kitaev transformation, or perhaps the Jordan-Wigner transformation). With more insight into the structure of the molecule, the outputs of these transformations, or the relations between the Pauli terms, a specific coloring algorithm could be chosen or designed for each application. In special cases, one could imagine that optimal colorings could be calculated efficiently.

%%%%%%%%%%%%%%%%
% BACK MATTER  %
%%%%%%%%%%%%%%%%

\appendix
\addcontentsline{toc}{chapter}{Appendices}
%======================================================================
\section{Diagonalization Algorithm}\label{app:A}
%======================================================================

We define the operators which act by conjugation on $q$-ary Pauli gates as follows \cite{Clifford}:
\begin{align*}
F_q :\ &X_q \mapsto Z_q \\
& Z_q \mapsto X_q^{-1} \\
R_q :\ & X_q \mapsto X_q Z_q \\
& Z_q \mapsto Z_q \\
\SUM_q :\ & \identity{q} X_q \mapsto \identity{q} X_q \\
& X_q \identity{q} \mapsto X_q X_q \\
& \identity{q} Z_q \mapsto Z_q^{-1} Z_q \\
& Z_q \identity{q} \mapsto Z_q \identity{q}.
\end{align*}
We shall use without proof the fact that these gates generate the generalized Clifford group, the set of operators which permutes the elements of $\pauli{q}$. We shall give an inductive proof which can be used as an iterative algorithm to construct the specific Clifford gate which simultaneously diagonalizes a set of commuting Pauli operators.

\begin{lemma}
Let $\mathcal{S}$ be a set of commuting Pauli operators. There exists a Clifford gate, $G$, such that $G \mathcal{S} G^\dagger = \{G P G^\dagger : P \in \mathcal{S}\} \subseteq \left\{ \bigotimes_{i=1}^{n} X_q^0 Z_q^{j_i} : j_i \in \Zmod{q} \right\}$.
\end{lemma}
\begin{proof}
We shall prove this by induction on the length of the Pauli operators, $n$.
\begin{itemize}
\item $n=1$: If $\mathcal{S} \subseteq \pauli{q}$ is not diagonalized, then there exists some operator, $P \in \mathcal{S}$, with a non-zero $X$-component. We shall write $P = X_q^a Z_q^b$ for some $a \not\equiv 0 \mod{q}$. Let:
\[
G_1 = R_q^{a^{-1}(q-b)}
\]
where $a^{-1}$ is the multiplicative inverse of $a$ over $\Zmod{q}$.

By our conjugation rules above, we observe that:
\begin{align*}
G_1 P G_1^\dagger &= X_q^a Z_q^b Z_q^{a(a^{-1}(q-b))} \\
&= X_q^a Z_q^b Z_q^{q-b} \\
&= X_q^a Z_q^q \\
&= X_q^a.
\end{align*}
We have successfully removed the $Z$-component of this Pauli operator, but we want to remove the $X$-component, so we let $G = F_q G_1$. Again, following our conjugation rules, we observe that:
\begin{align*}
G P G^\dagger &= F_q G_1 P G_1^\dagger F_q^\dagger \\
&= F_q X_q^a F_q^\dagger \\
&= Z_q^a.
\end{align*}
Thus, we have successfully diagonalized a single element of $\mathcal{S}$. However, since conjugation by Clifford gates preserves commutation relations, we know that $G \mathcal{S} G^\dagger$ must still pairwise commute. Since there exists an operator with no $X$-component on the one and only qudit, all operators in $G \mathcal{S} G^\dagger$ must have no $X$-component. Thus, we have successfully diagonalized the set.

\item $n > 1$: Assume we are able to successfully diagonalize any commuting set of Pauli operators on $n-1$ qudits. If $\mathcal{S} \subseteq \pauli{q}[n]$ is not already diagonalized, there exists some operator, $P \in \mathcal{S}$, with a nonzero $X$-component. Let $H_1$ be a $\SUM_q$ gate from a qudit with a nonzero $X$-component to the first qudit in our tensor product, we may now assume that $H_1 P H_1^\dagger = X_q^a Z_q^b \otimes P_1$ for some $a \not\equiv 0 \mod{q}$ and some $P' \in \pauli{q}[n-1]$.

We take advantage of our above proof and left-multiply our $H_1$ by the gate which acts only on the first qudit and diagonalizes it as in our base case (let's call this gate $H_2$). This leaves us with: $H_2 H_1 P H_1^\dagger H_2^\dagger = Z_q^a \otimes P_2$.

Next, $P_2$ may have some qudits with a nonzero $Z$-component. Let $H_3$ be a series of $\SUM_q$ gates from these qudits to the first qudit, applying the gate as many times as is necessary to cancel out the $Z$-component. This leaves us with: $H_3 H_2 H_1 P H_1^\dagger H_2^\dagger H_3^\dagger = Z_q^a \otimes P_3$.

$P_3$ has no $Z$-component, but my have some qudits with a nonzero $X$-component. We address this by letting $H_4$ be the $(n-1)$-fold $F_q$ gate applied to the last $n-1$ qudits, leaving us with: $H_4 H_3 H_2 H_1 P H_1^\dagger H_2^\dagger H_3^\dagger H_4^\dagger = Z_q^a \otimes P_4$.

$P_4$ has no $X$-component, but may have some qudits with nonzero $Z$-component. Using the same strategy as in constructing $H_3$, we use $\SUM_q$ gates to cancel out this $Z$-component, leaving us with: $H_5 H_4 H_3 H_2 H_1 P H_1^\dagger H_2^\dagger H_3^\dagger H_4^\dagger H_5^\dagger = Z_q^a \otimes \identity{q^{n-1}}$.

Letting $G_1$ denote this product of Clifford gates we have constructed, we observe that every operator in $G_1 \mathcal{S} G_1^\dagger$ must commute with $Z_q^a \otimes \identity{q^{n-1}}$. This means that every operator in $G_1 \mathcal{S} G_1^\dagger$ must have no $X$-component in the first qudit.

Restricting ourselves to Clifford group operations on the last $n-1$ qudits, we already know by our inductive hypothesis that there exists a Clifford gate, $G_2$, which simultaneously diagonalizes the remaining $n-1$ qudits of $G_1 \mathcal{S} G_1^\dagger$, without affecting the first qudit.

In conclusion, every operator in $G_2 G_1 \mathcal{S} G_1^\dagger G_2$ has no $X$-component on every qudit. Thus, $G = G_2 G_1$ is a Clifford gate which diagonalizes $\mathcal{S}$.
\end{itemize}
\end{proof}

\phantomsection  % With hyperref package, enables hyperlinking from the table of contents to bibliography             
% The following statement causes the title "References" to be used for the bibliography section:
\renewcommand*{\bibname}{References}

% Add the References to the Table of Contents
\addcontentsline{toc}{section}{\textbf{References}}

\bibliographystyle{unsrtnat}
\bibliography{Bibliography}

% The following statement causes the specified references to be added to the bibliography% even if they were not 
% cited in the text. The asterisk is a wildcard that causes all entries in the bibliographic database to be included (optional).
\nocite{*}

\end{document}